\documentclass[
submission
, nomarks
]{dmtcs-episciences}

\usepackage[utf8]{inputenc}

\pdfoutput=1

\usepackage[numbers]{natbib}

\usepackage{bbm}
\usepackage{mathrsfs}
\usepackage{latexsym}
\usepackage{microtype}

\usepackage{suffix}
\usepackage{caption}
\usepackage{wrapfig}

\usepackage{xpunctuate}
\usepackage[all,british]{foreign} % Typesets abbreviations like e.g.
\redefnotforeign[ie]{i.e\xperiodafter}
\redefnotforeign[eg]{e.g\xperiodafter}

\usepackage{enumitem}

\usepackage{marginnote}

\usepackage[draft,english,silent]{fixme}
\fxsetup{theme=color,layout=marginnote,innerlayout=noinline}

\usepackage{amsmath, amssymb, amsfonts, mathtools} % AMS math
\usepackage{amsthm}
\usepackage[full,small]{complexity}

\newtheorem{theorem}{Theorem}[section]
\newtheorem{corollary}{Corollary}[theorem]
\newtheorem{lemma}[theorem]{Lemma}
\newtheorem{definition}[theorem]{Definition}
\newtheorem*{claim}{Claim}

\renewcommand{\qed}{\hfill $\square$} 
\PassOptionsToPackage{dvipsnames,usenames,table}{xcolor}
\usepackage{tikz}
\usetikzlibrary{calc,matrix,fit,patterns,decorations.pathreplacing,arrows,shapes}
\usepackage{tikz-layers}

\usepackage{float} % provides the [H] "exactly here" float placement specifier
\usepackage[longend,vlined,plain]{algorithm2e}

\renewcommand{\G}{\mathbb G}
\newcommand{\X}{\mathbb X}

\newcommand{\N}{\mathbb N}
\renewcommand{\H}{\mathbb H}
\renewcommand{\K}{\mathbb K}
\renewcommand{\S}{\mathbb S}

\newcommand{\subgraph}{\subseteq}
\newcommand{\Count}{\#}

\DeclareMathOperator{\im}{im}
\DeclareMathOperator{\rep}{rep}
\DeclareMathOperator{\aut}{aut}
\DeclareMathOperator{\orb}{orb}
\DeclareMathOperator{\supp}{supp}

\DeclareMathOperator{\jointex}{J}

\newcommand{\ix}{\iota}
\newcommand{\defeq}{\coloneqq}
\newcommand{\any}{\mathord{\color{black!33}\bullet}}

\newcommand{\supply}{\operatorname{\mathsf{supply}}}
\newcommand{\demand}{\operatorname{\mathsf{demand}}}

\newbox\mytempbox
\tikzstyle{embeds} = [->, >=open triangle 45]
\newcommand{\embeds}[1]{%
	\sbox{\mytempbox}{\hbox{\( \scriptstyle\mkern14mu#1\mkern14mu \)}}
	\tikz[baseline=-0.5ex]{%
		\draw[right hook->] (0,0) -- node[midway,above=-0.3ex]{\usebox\mytempbox} (\wd\mytempbox,0);%
	}%
}
\newcommand{\embedsweak}[1]{%
	\sbox{\mytempbox}{\hbox{\( \scriptstyle\mkern14mu#1\mkern14mu \)}}
	\tikz[baseline=-0.5ex]{%
		\draw[right hook->, densely dotted] (0,0) -- node[midway,above=-0.3ex]{\usebox\mytempbox} (\wd\mytempbox,0);%
	}%
}

\definecolor{amazing}{RGB}{254,67,101}
\definecolor{cardinal}{HTML}{BB333C}
\definecolor{niceblack}{HTML}{020300}

\tikzset{%
	vertex/.style={circle,fill=niceblack!15,minimum size=18pt,inner sep=0pt},
	red edge/.style={draw,thick,-,cardinal!80},
	black edge/.style={draw,line width=.8pt,-,niceblack},
	gray edge/.style={draw,line width=.8pt,-,niceblack!15},    
	box/.style={fill,cardinal,inner sep=5pt,rounded corners=15pt}
}

\SetKwInput{KwInput}{Input}
\SetKwInput{KwOutput}{Output}
\SetKw{Continue}{continue}

\SetCommentSty{mycommfont}

\newcommand*\circled[4]{\tikz[baseline=(char.base)]{
		\node[circle,fill=#3,draw=#4,inner sep=1pt,minimum size=1.2em] (char) {\color{#2} #1};}}

\def\commentmark#1{\circled{#1}{gray}{white}{gray}}
\def\comment#1#2{\textcolor{gray}{\commentmark{#1} #2}}

\usepackage{environ}
\newcommand{\repeattheorem}[1]{%
	\begingroup
	\renewcommand{\thetheorem}{\ref{#1}}%
	\expandafter\expandafter\expandafter\theorem
	\csname reptheorem@#1\endcsname
	\endtheorem
	\endgroup
}
\NewEnviron{reptheorem}[1]{%
	\global\expandafter\xdef\csname reptheorem@#1\endcsname{%
		\unexpanded\expandafter{\BODY}%
	}%
	\expandafter\theorem\BODY\unskip\label{#1}\endtheorem
}

\newcommand{\repeatlemma}[1]{%
	\begingroup
	\renewcommand{\thelemma}{\ref{#1}}%
	\expandafter\expandafter\expandafter\lemma
	\csname replemma@#1\endcsname
	\endtheorem
	\endgroup
}
\NewEnviron{replemma}[1]{%
	\global\expandafter\xdef\csname replemma@#1\endcsname{%
		\unexpanded\expandafter{\BODY}%
	}%
	\expandafter\lemma\BODY\unskip\label{#1}\endtheorem
}

\newcommand*\varrule[1][0.4pt]{\leavevmode\leaders\hrule height#1\hfill\kern0pt}

\title{Counting large patterns in degenerate graphs}

\author{%
	Christine Awofeso \and  % \orcidID{0009-0000-3550-1727} 
	Patrick Greaves \and % \orcidID{0009-0007-0752-0526} 
	Oded Lachish \and % \orcidID{0000-0001-5406-8121}
	Felix Reidl \and  % \orcidID{0000-0002-2354-3003}
}
\affiliation{%
  Birkbeck College, University of London, UK
}

\begin{document}

% \publicationdata
% {vol. 25:3 special issue for main purpose}
% {2022}
% {1}
% {10.46298/dmtcs.10472}
% %{1998-10-14; 1998-10-14; 2002-07-19; 2014-02-05; 2015-09-09; 2022-12-25}
% {1999-01-01}
% {1999-01-01}
% {1999-01-01}

\maketitle
\begin{abstract}
	Subgraph counting is a fundamental algorithmic problem with many applications, including in the analysis of social and biological networks. The problem asks for the number of occurrences of a pattern graph $H$ as a subgraph of a host graph $G$ and is known to
	be computationally challenging: it is $\#W[1]$-hard even when $H$ is restricted to simple structures such as cliques or paths.
	Curticapean and Marx~(FOCS'14) show that if the graph $H$ has vertex cover number $\tau$, subgraph counting has time complexity $O(|H|^{2^{O(\tau)}} |G|^{\tau + O(1)})$.
	
	This raises the question of whether this upper bound can be improved for input graphs $G$ from a restricted family of graphs. Earlier work by Eppstein~(IPL'94) shows that this is indeed possible, by proving that when $G$ is a $d$-degenerate graph and $H$ is a biclique of arbitrary size, subgraph counting has time complexity $O(d 3^{d/3} |G|)$. 
		
 	We show that if the input is restricted to $d$-degenerate graphs, the upper bound of Curticapean and Marx can be improved for a family of graphs $H$ that includes all bicliques and satisfies a property we call \emph{$(c,d)$-locatable}. Importantly, our algorithm's running time only has a polynomial dependence on the size of~$H$.

	A key feature of $(c,d)$-locatable graphs $H$ is that they admit a vertex cover of size at most $cd$. We further characterise $(1,d)$-locatable graphs, for which our algorithms achieve a linear running time dependence on $|G|$, and we establish a lower bound showing that counting graphs which are barely \emph{not} $(1,d)$-locatable is already $\#W[1]$-hard.	

    We note that the restriction to $d$-degenerate graphs has been a fruitful line of research leading to
    two very general results (FOCS'21, SODA'25), and this creates the impression that we largely understand the complexity of counting substructures in degenerate graphs.
    However, all aforementioned results have an exponential dependency on the size of the pattern graph~$H$. 
\end{abstract}

\section{Introduction}\label{sec:Intro}

Recall that a graph is $d$-degenerate iff its vertices can be ordered in such a way that every vertex has at most~$d$ neighbours preceding it in the order. 
Bressan and Roth~\cite{bressanDichotomoy2022} recently proved several lower and upper bounds related to counting substructures---subgraphs, induced subgraphs, homomorphisms---in degenerate graphs which, up to some minor factors, provide a characterisation of these problems:
Given a graph~$H$ and a $d$-degenerate graph~$G$, they show that counting
how often~$H$ appears as a subgraph in~$G$ is possible in expected time
$f(|H|,d) |G|^{\im(H)}$ and counting how often it appears as an induced subgraph
in time $f(|H|,d) |G|^{\alpha(H)}$, where $\im(H)$ is the induced matching number and~$\alpha(H)$ the independence number of~$H$. Assuming the exponential time hypothesis (ETH), Bressan and Roth further show almost matching lower bounds%
\footnote{
$f(|H|,d) |G|^{o(\im(H) / \log \im(H))}$ and $f(|H|,d) |G|^{o(\alpha(H)/\log \alpha(H))}$, respectively
}. Previous work had focused on hardness classifications for counting substructures without any degeneracy restrictions on~$G$ (\cite{curticapeanMarxCount2014,chen2008icalp,dalmau2004tcs}), parameterising the treewidth of~$G$ (\cite{marx2010toc}), and the generalised subgraph problem (\cite{roth2024siamcomp}).

Regarding the \emph{linear-time} countability of subgraphs in degenerate graphs,  Bera \etal~\cite{beraCountingSix2020} proved that all graphs~$H$ on at most five vertices can be counted in time~$O(\|G\| d^{|H|-2})$. Moreover, they showed that counting cycles on six or more vertices is not possible in linear time unless the Triangle Detection Conjecture\footnote{
    The conjecture states that there exists a constant~$\delta$ such that
    any algorithm under the word-RAM model takes time at least
    $\|G\|^{1+\delta-o(1)}$ in expectation to detect whether~$G$ contains a triangle~\cite{abboudTriangleConjecture2014}.
} (TDC) fails. Bressan~\cite{bressanDAGwidth2021} designed an algorithm that counts subgraphs in time~$O(f(d,|H|) |G|^{\tau_1(H)})$, where $\tau_1(H)$ is a width-measure called the \emph{DAG treewidth} of~$H$. Regarding lower bounds, Bera \etal~\cite{beraCountingDegenerate2022} provided a complete characterisation of the graphs~$H$ that can be counted in linear time in a degenerate host graph as a subgraph/induced subgraph under the TDC\footnote{The exact characterisations are somewhat technical, but they revolve round the existence of~$C_k$, $k \geq 6$, as certain types of substructures in~$H$}. 

These results create the impression that we largely understand the complexity of counting substructures in degenerate graphs. The problem is, in all aforementioned upper bounds the running times have an exponential dependence on~$|H|$. 
However, in the general case (\cf Theorem 2.8 in~\cite{curticapeanMarxCount2014}), if $H$ has vertex cover number $\tau$, subgraph counting can be solved in time $O(|H|^{2^{O(\tau)}} |G|^{\tau + O(1)})$, and much earlier Eppstein~\cite{eppsteinBicliques1994} showed that when $G$ is a $d$-degenerate graph and $H$ is a biclique, subgraph counting can be solved in time $O(d 3^{d/3} |G|)$. 
This leaves the question of whether the upper bound of Curticapean and Marx can be improved for a larger family of graphs than just bicliques.

\vspace*{-6pt}
\paragraph{Our contribution:}
We show that if the input is restricted to $d$-degenerate graphs, the upper bound of Curticapean and Marx can be improved for a family of graphs $H$ satisfying a property we call \emph{$(c,d)$-locatable}, which includes all bicliques.
Specifically, we show that for every such graph~$H$, the number of times $H$ occurs as a (non-induced) subgraph of~$G$ can be computed
in time~$O\big(cd^2 f(c,d) |H|^{f(c,d)} |G|^c \big)$ where~$f(c,d) = 2^d(cd)^{d+1}$ (Theorem~\ref{thm:count-weak}, Section~\ref{sec:algorithm}), and as a semi-induced subgraph%
\footnote{A pattern $H$ with vertex cover~$S$ and independent set~$T$ appears \emph{semi-induced} in~$G$ if some supergraph~$\hat H$ obtained by only
	adding edges between $T$-vertices in~$H$ appears as an induced subgraph of~$G$.} %
in time $O\big((cd)^{2cd+d} |G|^c \big)$ (Theorem~\ref{thm:count-strong}, Section~\ref{sec:algorithm}). If the host graph~$G$ is bipartite, the same time complexity bounds hold for counting induced subgraphs.
We note that our result is incomparable to Bressan's algorithm, since our algorithm can count patterns with large independent sets and unbounded DAG treewidth.

Both ideas generalise the classical approach for counting bicliques recalled at the start of Section~\ref{sec:algorithm}: a biclique's small side is \emph{located} by guessing a single host-graph vertex and searching only its left-neighbourhood, and then counted via a simple 
joint-neighbourhood argument that works because every vertex on the large side is adjacent to \emph{all} of the small side. \emph{Locatability} is the generalisation of the search step: a $(c,d)$-locatable pattern's small side~$S$ can still be found inside the left-neighbourhood 
of some small witness set of size~$c$, so we only ever need to search~$|G|^c$ candidate sets rather than all of~$G$. General patterns break the counting step, however, since their large side need not be adjacent to all of~$S$ uniformly. We fix this with the \emph{index representation} 
which is a compact combinatorial signature recording exactly how the large side attaches to the small side. We enumerate the few signatures realisable in a $d$-degenerate graph and count occurrences separately per signature, using the pattern's automorphisms to avoid over- or under-counting.

To compare these results to those of Curticapean and Marx, note that the~$cd$ in our result upper-bounds the vertex cover number~$\tau$ of~$H$, \ie~$\tau \le cd$. In the best case, where this bound is tight (\ie~$cd = \tau$), expressing our running time in terms of~$\tau$ gives that
$O\big(cd^2 f(c,d) |H|^{f(c,d)} |G|^c \big)$ becomes
$O\big(\tau d \cdot f(\tau,d) |H|^{f(\tau,d)} {|G|^{\tau/d}} \big)$ where~$f(\tau,d) = 2^d\tau^{d+1}$, and
$O\big((cd)^{2cd+d} |G|^c \big)$ becomes $O\big(\tau^{2\tau+d} {|G|^{\tau/d}} \big)$; in both cases this improves the exponent of~$|G|$ by a factor of~$1/d$ over Curticapean and Marx's bound in terms of the same parameter~$\tau$.

We are specifically interested in~$(1,d)$-locatable patterns as these can be counted with only a linear dependence on~$|G|$ in the running time.
We characterise these patterns (Section~\ref{sec:structure}) and supplement our positive algorithmic results with a lower bound that shows that even counting patterns very close to being~$(1,d)$-locatable is already $\#W[1]$-hard (Section~\ref{sec:lower-bound}). Note that we explicitly link the degeneracy~$d$ to the pattern via the locatability property, so our setting differs from the characterisation by Bera \etal~\cite{beraCountingDegenerate2022}, who ask which patterns can be counted in linear time for \emph{all}~$d$. This means that~$(1,d)$-locatable patterns can be counted in linear time in graphs of degeneracy~$\leq d$, but might not be in graphs of degeneracy~$> d$.

\section{Preliminaries}\label{sec:prelims}

\marginnote{$\any$}
We will occasionally use the placeholder $\any$ for variables whose value is irrelevant in the given context.

\marginnote{$[\any]$, $\X$, $<_\X$, $\max_\X$, $\min_\X$, $\supp(\any)$}%
\noindent
For an integer $k$, we use $[k]$ as a short-hand for the set $\{0, 1, 2, \ldots, k-1\}$. We use blackboard bold letters like~$\X$ to denote totally ordered sets, that is, some underlying set~$X$ associated with a 
total order~$<_\X$. We further, for~$Y \subseteq X$, use the notations~$\max_\X Y$ to mean the maximum member in~$Y$ under~$<_\X$ and the similarly defined~$\min_\X Y$. If the context allows it, we will sometimes drop this subscript, \eg we shorten $\max_\X X_1 <_\X \max_\X X_2$ to~$\max X_1 <_\X \max X_2$ or~$\max_\X \X$ to simply~$\max \X$. Finally, for a real-valued function~$f$ we write~$\supp(f)$ to denote its support, \ie the set of all values for which~$f$ is non-zero. 

\marginnote{$\ix_\X$}
The 
\emph{index function}~$\ix_\X \colon X \to \N$ maps elements of~$X$ to their
corresponding position in~$\X$, where~$\X$ is the set $X$ imbued with some linear order. We extend this function to sets via
$\ix_\X(S) = \{ \ix_\X(s) \mid s \in S \}$. For any integer~$i \in [|X|]$ we 
write $\X[i]$ to mean the $i$th element in the ordered set.
An \emph{index set}~$I$ for~$\X$ is simply a subset of~$[|X|]$ and we extend 
the index notation to sets via~$\X[I] \defeq \{ \X[i] \mid i \in I \}$.

\marginnote{$|G|$, $\|G\|$, $\delta$, $\aut$}
For a graph $G$ we use $V(G)$ and $E(G)$ to refer to its vertex- and edge-set,
respectively, and $\delta(G)$ to denote its minimum degree. We use the shorthands $|G| \defeq |V(G)|$ and $\|G\| \defeq |E(G)|$. We write $\aut(G)$ for the number of automorphisms of~$G$.

\marginnote{$\G$, ordered graph}
An \emph{ordered graph} is a pair $\G = (G, <)$ where $G$ is a graph and $<$ a
total ordering of $V(G)$. We write $<_\G$ to denote the ordering for a given
ordered graph and extend this notation to the derived relations $\leq_\G$,
$>_\G$, $\geq_\G$. For simplicity we will call $\G$ an \emph{ordering of} $G$.

\marginnote{$N^-, \Delta^-$}
We use the same notations for graphs and ordered graphs, additionally we write
$N_\G^-(u) \defeq \{ v \in N(u) \mid v <_\G u \}$ for the \emph{left-neighbourhood} of a vertex $u \in \G$. We write $N_\G^-[u] \defeq N_\G^-(u) \cup \{u\}$ for the closed left-neighbourhood which we extend to vertex sets~$X$ via $N_\G^-[X] \defeq \bigcup_{u \in X} N_\G^-[u]$. We write~$\Delta^-(\G) \defeq \max_{u \in \G} |N_\G^-(u)|$ to denote the maximum left-degree of an ordered graph.

% We further use $d_\G^-(u)$ and
% $d_G^+(u)$ for the left and right degree, as well as $\Delta^-(\G) \defeq \max_{u
% \in \G} d_\G^-(u)$ and $\Delta^+(\G) \defeq \max_{u \in \G} d_\G^+(u)$. We omit the graphs in the subscripts if clear from the context.

\marginnote{isomorphic}
Two ordered graphs~$\G_1, \G_2$ are \emph{isomorphic} if the function that maps the $i$th vertex of~$\G_1$ to the $i$th vertex of~$\G_2$ is a graph isomorphism. We write $\G_1 \simeq \G_2$ to indicate that two ordered graphs are isomorphic.

\marginnote{degeneracy}
A graph~$G$ is \emph{$d$-degenerate} if there exists an ordering $\G$ such that
$\Delta^-(\G) \leq d$. 
We say an ordering $\G$ of a graph $G$ is $d$-degenerate if   $\Delta^-(\G) \leq d$, and just \emph{degenerate} if $\Delta^-(\G)$ is minimum over all orders of $G$. The degeneracy ordering of a graph can be computed in time $O(|G| + \|G\|)$ in general~\cite{matulaDegeneracy1983}, and therefore $O(d|G|)$ for $d$-degenerate graphs.

\vspace*{-6pt}
\paragraph*{The degeneracy toolkit.}

We will make use of the following data structures for degenerate graphs by Drange \etal:
\begin{lemma}[\cite{drangeComplexityDegenerate2023}]\label{lemma:R}
  Let $\G$ be an ordered graph with degeneracy $d$. Then in time $O(d2^d |G|)$
  we can compute a subset dictionary $R$ over~$V(G)$ which for any $X \subseteq
  V(G)$ answers the query
  $
    R[X] \defeq \big| \{ v \in G \mid X \subseteq N^-(v) \} \big|
  $
  in time~$O(|X|)$.
\end{lemma}

\begin{theorem}[\cite{drangeComplexityDegenerate2023}]\label{thm:data-structure}
  Let $\G$ be an ordered graph with degeneracy~$d$. After a
  preprocessing time of~$O(d2^d |G|)$, we can, for any given~$S \subseteq V(G)$,
  compute a subset dictionary~$Q_S$ in time~$O(|S| 2^{|S|} + d|S|^2)$ which
  for any $X \subseteq S \subseteq V(G)$ answers the query
  $
    Q_S[X] \defeq \big| \{ v \in V(G)\setminus X \mid S \cap N(v) = X \} \big|
  $
  in time~$O(|X|)$.\footnote{The data structure in the referenced paper counts all vertices $v \in V(G)$ with the given property but it is trivial to exclude vertices contained in~$X$.}
\end{theorem}

\noindent\marginnote{neighbourhood class}
In the following we will need the concept of a \emph{neighbourhood class}. For a bipartite graph with sides~$S$, $T$ we partition~$T$ into sets such that
$u,v \in T$ are in the same class iff $N(u) = N(v)$, that is, $u$ and~$v$ have exactly the same neighbours in~$S$. 

In the following we will make use of this simple observation about degenerate graphs and neighbourhood classes:
\begin{lemma}\label{lemma:neighbourhood-classes}
    Let~$H = (S,T,E)$ be a bipartite $d$-degenerate graph. Then the number
    of vertices in~$T$ of degree at least~$d+1$ is at most~$(d+1)|S|$ and the number
    of neighbourhood classes in~$T$ is at most~$4|S|^d$.
\end{lemma}
\begin{proof}
    Let~$\H$ be a $d$-degenerate ordering of~$H$ and let $T_{>d} \defeq \{ t \in T \mid \deg_H(t) > d \}$.
    Given every~$t \in T_{>d}$ has all its neighbours in~$S$, yet  at most~$d$ of them to its left; 
    it must be the case that some~$s \in S$ satisfies~$t <_\H s$.
    This means that~$t \in N^-_\H(s) \subseteq N^-_\H[S]$ and thus~$T_{>d} \subseteq N^-_\H[S]$.
    By definition,~$N^-_\H[S]$ is the union of the~$|S|$ closed left-neighbourhoods
    $N^-_\H[s]$, $s \in S$, and each has size at most~$d+1$ by $d$-degeneracy. We can then 
    bound~$|N^-_\H[S]| \leq (d+1)|S|$, proving the first statement.

    Let~$T_{\leq d} \defeq T \setminus T_{>d}$. If~$d=0$, then~$H$ is edgeless, so~$T$ has
    exactly one class; if~$d=1$, then~$T_{\leq1}$ has exactly~$|S|+1$ classes, giving at
    most~$3|S|+1\leq4|S|$ classes together with~$|T_{>1}|\leq2|S|$.

    For~$d\geq2$: if~$d<|S|$, each vertex of~$T_{\leq d}$ has at most~$d$ neighbours in~$S$,
    so the entropy bound gives at most~$\sum_{i=1}^d\binom{|S|}{i}\leq\big(\frac{e|S|}{d}\big)^d
    \leq3|S|^d$ classes among~$T_{\leq d}$, and one checks~$(d+1)|S|\leq|S|^d$; if~$d\geq|S|$
    instead, there are trivially at most~$2^{|S|}\leq|S|^d$ classes among~$T_{\leq d}$, and
    $(d+1)|S|\leq3|S|^d$. Either way, combined with the bound on~$T_{>d}$, this gives at
    most~$4|S|^d$ classes in total.

    Since for~$|S|=1$ there is exactly one neighbourhood
    class, the bound holds for all sizes of~$S$ and the second claim follows. \qed
\end{proof}

\noindent
We note that the neighbourhood classes of a bipartite graph~$H$ can be computed in time~$O(|H| + \|H\|)$ using the partition-refinement data structure (\cite{habibPartitionRefinement1998}).

\section{Counting large patterns}

Let us review how bicliques of arbitrary size can be counted in linear time in degenerate graphs. There are three ideas to this algorithm: First, bicliques with two large sides cannot exist in degenerate graphs. Second, if a biclique has at least one small side, then this side can be located in the left-neighbourhood of a vertex and therefore can be found in linear time. Third, once we have found one side of a biclique we can use a suitable data structure to count how many joint neighbours this side has in the host graph and therefore count the number of bicliques `rooted' at this vertex set.

\begin{theorem}[adapted from~\cite{eppsteinBicliques1994}]
    There exists an algorithm that for any $s \leq t \in \N$ and
    any $d$-degenerate graph $G$ counts the number of (non-induced)
    $K_{s,t}$ in~$G$ in time $O(d^2 2^d |G|)$.
\end{theorem}
\begin{proof}
    We begin with the simple case of $s \geq d+1$, in which case $K_{s,t}$
    is not $d$-degenerate and therefore cannot occur in $G$. In all remaining cases, our algorithm computes a $d$-degeneracy ordering $\G$
    of~$G$ in time~$O(d |G|)$ and preprocesses the graph in time $O(d2^d |G|)$ to compute the data structure~$R$ described in Lemma~\ref{lemma:R}.

    Next, consider the regime where $s \leq d$ and $t \geq d+1$. Then for any occurrence of $K_{s,t}$ as $\K \subgraph \G$ it holds that the last
    vertex $x = \max \K$ must be part of the $t$-side of the $K_{s,t}$ since otherwise we would find a vertex with left-degree $t > d$. We can therefore locate the $s$-side of the $K_{s,t}$ by guessing $x$
    and then guessing an $s$-sized subset $S \subseteq N^-[x]$ in total time $O(|G| {d \choose s}) = O(2^d |G|)$. Let~$c_l$ denote the number of joint
    neighbours of~$S$ that appear not to the right of~$\max_\G S$, \ie
    $
        c_l \defeq \Big| \bigcap_{u \in S} N_G(u) \cap \{ v \mid v < \max_\G S \} \Big|
    $,
    where~$\bigcap_{u \in S} N_G(u)$ is the set of joint (common) neighbours of~$S$.
    Note that we can compute $c_l$ by simply inspecting each of the $\leq d$ vertices in $N^-(\max_\G S)$ in time~$O(ds) = O(d^2)$. Let
    $c_r$ denote the number of joint neighbours of~$S$ that appear to the
    right of $\max_\G S$ and note that $c_r = R[S]$, meaning we can retrieve it in time $O(s) = O(d)$. Given $c_l$ and $c_r$, the number of $K_{s,t}$ in~$G$ whose $s$-side is exactly $S$ is then given by
    ${c_l + c_r \choose t}$.

    Finally, consider the regime where $s \leq t \leq d$. Let $\K \subgraph \G$ be any occurrence of $K_{s,t}$ with sides $S$, $T$. Let $x \defeq \max \K$, then either $S \subseteq N^-(x)$ or $T \subseteq N^-(x)$. In either case, we can proceed exactly as in the previous case to count the occurrence of \emph{all} $K_{s,t}$ with this specific $S$- or $T$-set in $G$.
    \qed
\end{proof}

\noindent
Based on these core ideas we now develop our algorithm to count \emph{large} graphs within a degenerate host graph. These graphs will have a specific structure, which we formalise here:

\begin{definition}[Pattern]\marginnote{Pattern, $S$, $T$}
    A \emph{pattern} is a connected graph $H$ whose vertex set is partitioned
    into sets $S,T$, called its \emph{sides}, where the set $T$ is independent in $H$, $T \subseteq N(S)$ and $|S| < |T|$.
    We call $S$ the \emph{small} and $T$ the \emph{large} side of the pattern.
\end{definition}

\noindent
We will look at two variants of how patterns might appear in a host graph. The first corresponds to a `semi-induced' subgraph, meaning that all edges and non-edges inside~$S$ and between~$S$ and~$T$ must appear exactly as in the pattern. However, we allow for arbitrary edges to appear between $T$-vertices.

\begin{definition}[Strong pattern containment, -- frequency]%
    \marginnote{$H \embeds{\phi} G$, $\Count(H \embeds{} G)$}%
    We say that a pattern $H$ with sides $S$, $T$ is \emph{strongly contained} in another graph $G$ if there exist disjoint
    vertex sets $\tilde S,\tilde T \subseteq V(G)$ and a bijection $\phi\colon S 
    \cup T \to \tilde S \cup \tilde T$ with the following properties:
    \vspace*{-6pt}
    \begin{itemize}
        \item $\phi(S) = \tilde S$ and $\phi(T) = \tilde T$, and
        \item for every pair $u \in S$, $v \in S \cup T$ it holds 
        that $uv \in H \iff \phi(u)\phi(v) \in G$.
    \end{itemize}
    \vspace*{-12pt}
    We call $\phi$ a \emph{strong embedding} of $H$
    into~$G$ and use the shorthand~$H \embeds{\phi} G$ to denote this fact. We further define $\Count(H \embeds{} G)$ as the number of strong embeddings of~$H$ into~$G$.
\end{definition}

\noindent
The second notion of pattern containment corresponds to the standard notion of a subgraph. Given that the strong containment above required
a bespoke bijection-based definition to capture semi-induced subgraphs, we phrase this weak containment in the same style to present
a unified approach:

\begin{definition}[Weak pattern containment, -- frequency]%
    \marginnote{$H \embedsweak{\psi} G$, $\Count(H \embedsweak{} G)$}%
    We say that a pattern $H$ with sides $S$, $T$ is \emph{weakly contained} in another graph $G$ if there exist disjoint
    vertex sets $\tilde S,\tilde T \subseteq V(G)$, an edge
    set $\tilde F \subseteq E(G[\tilde S\cup \tilde T])$, and a bijection $\psi\colon V(H) \cup E(H) \to \tilde S \cup \tilde T \cup \tilde F$ with the following properties:
    \vspace*{-6pt}
    \begin{itemize}
        \item $\psi(S) = \tilde S$ and $\psi(T) = \tilde T$, and
        \item for every edge $uv \in H$ it holds that
        $\psi(u)\psi(v) = \psi(uv)$ and
        $\psi(uv) \in \tilde F$.
    \end{itemize}
    \vspace*{-12pt}
    We call $\psi$ a \emph{weak embedding} of $H$
    into~$G$ and use the shorthand~$H \embedsweak{\psi} G$ to denote this fact. We further define $\Count(H \embedsweak{} G)$ as the number of weak embeddings of~$H$ into~$G$.
\end{definition}

\noindent
We record the edge map~$\tilde F = \psi(E(H))$ explicitly, rather than leaving it implicit in the vertex map, because weak containment only asks for a (non-induced) subgraph: 
the host graph~$G[\tilde S \cup \tilde T]$ may contain edges beyond~$\tilde F$, so the vertex images~$\tilde S, \tilde T$ alone would not pin down which specific edge set 
constitutes this particular occurrence of~$H$. Since~$\psi$ therefore maps edges as well as vertices, we will use the notation~$\psi(H)$ to denote the subgraph of~$G$ onto which~$H$ is mapped by~$\psi$.

\marginnote{\footnotesize$\Count\!(\any \!\embeds{\!\!\phi\!\!}\! \any \!\!\mid\!\! P(\phi))$}
In the following we will often count embeddings with certain additional properties. In these cases, we will use the notation 
$
    \Count(H \smash{\embeds{\phi}} G \mid P(\phi)) \defeq |\{ H \smash{\embeds{\phi}} G \mid P(\phi) \}|
$,
where~$P$ is some predicate about~$\phi$, to mean the number of strong embeddings~$\phi$ for which~$P$ holds. We use an analogous notation for weak embeddings.

\subsection{Computing and encoding pattern automorphisms}\label{sec:automorphisms}

In contrast to the nicely symmetric bicliques, arbitrary patterns 
have much more complicated automorphisms that we have to compute explicitly. Since~$S$ is a vertex cover of the pattern, the automorphisms of~$H$ are 
pretty much fixed by the orbit of~$S$ under the automorphism action. In particular, they are very tame if the orbit of~$S$ only contains~$S$ itself. However, this is not true for all patterns and
thus we have to explicitly compute the orbit in order to proceed.
\marginnote{$\orb_S$}
The $S$-orbit $\orb_S(H)$ contains all sets~$X \subseteq V(H)$ onto which~$S$ is mapped by some automorphism of~$H$, that is
$
    \orb_S(H) := \{ X \subseteq V(H) \mid \exists\text{ automorphism}~\lambda~\text{of~$H$ with}~\lambda(S) = X\}
$.

\begin{lemma}\label{lemma:S-orbit}
    For every pattern~$H$ the $S$-orbit $\orb_S(H)$ can be enumerated
    in time~$O(|S|^{|S|} \cdot |H|^4 )$. If~$H$
    is $d$-degenerate, the running time improves to $O(d^2 |S|^{|S|+d} \cdot |H|^2)$.
\end{lemma}
\begin{proof}
	Let us partition the set~$T$ according to neighbourhoods in~$S$: let~$\{T_X\}_{X \subseteq S}$ be the neighbourhood classes with~$T_X \defeq \{ u \in T \mid N_H(u) = X \}$.
	
	We make use of the fact that~$S$ is a vertex cover of~$H$, and therefore
	any automorphism of~$H$ must map~$S$ onto another vertex cover of size~$|S|$. We enumerate all such vertex covers in time~$O(2^{|S|} \cdot |H|)$ using the classical branching algorithm. For each such vertex cover~$C$, we enumerate all $|S|!$ mappings~$\sigma\colon S \to C$.
	Note that checking whether a given~$\sigma$ can be extended to an automorphism is simple: first compute the neighbourhood classes of~$\{ U_X \}_{X \subseteq C}$ of~$V(H)\setminus C$ in~$H$ with respect to~$C$ in time~$O(|H| + \|H\|)$ using partition refinement~\cite{habibPartitionRefinement1998}.
	We now check whether~$\sigma(S)$ has the same neighbourhood structure as~$S$, meaning that~$|T_X| = |U_{\sigma(X)}|$ for all non-empty neighbourhood classes~$T_X$. Since the number of classes is bounded by~$|T|$, this can be done in time~$O(\|H\| \cdot |S|)$ for general patterns using a suitable set data structure to store~$\sigma$. We arrive at a running time of
	$
	O(2^{|S|} |H| (|H|+\|H\|) |S|! \cdot \|H\|^2 |S| )
	= O(|S|^{|S|} |H|^4 ),
	$
	where we used that~$\|H\| = O(|S| \cdot |H|)$ and that~$k^4 2^k k! = O(k^k)$.
	
	If~$H$ is~$d$-degenerate, we can apply Lemma~\ref{lemma:neighbourhood-classes} to bound the number of non-empty neighbourhood classes of~$T$ by $4|S|^d$.
	We can therefore replace the factors~$O(\|H\|)$ in the above argument with~$O(|S|^d)$ after an initial preprocessing of time~$O(|H| + \|H\|) = O(d|H|)$ to compute the classes.
	% \begin{gather*}
		%     O(2^{|S|} (|H|+\|H\|) |S|! \cdot \|H\| \cdot 4|S|^d \cdot |S| ) \\
		%     = O(2^{|S|} (d+1)|H| |S|! \cdot d |H| \cdot |S|^{d+1} ) \\
		%     = O(d^2 |H|^2 \cdot 2^{|S|} |S|! |S|^{d+1} ) \\
		%     = O(d^2 |H|^2 \cdot |S|^{|S|+d} ) \\
		% \end{gather*}
	\qed
\end{proof}

\noindent%
\marginnote{$\aut_S$}%
The second quantity we will need are automorphisms that stabilise~$S$, which we will denote by $\aut_S(H)$. That is, $\aut_S(H)$ contains all automorphisms~$\lambda \in \aut(H)$ with~$\lambda(S) = S$.

\begin{lemma}\label{lemma:count-S-autos}%
    For every pattern~$H$ the number of automorphisms $\aut_S(H)$ can be computed in time~$O(|S|^{|S|} \cdot |H|)$.
    If~$H$ is $d$-degenerate, the running time improves to~$O(\|H\| + |S|^{|S| + d})$.
\end{lemma}
\begin{proof}
	Let again~$\{T_X\}_{X \subseteq S}$ be the neighbourhood classes of~$T$
	which we can construct in time~$O(\|H\|)$ using partition refinement~\cite{habibPartitionRefinement1998}. For every
	mapping~$\sigma\colon S \to S$ which preserves the neighbourhood
	structure, \ie $|T_X| = |T_{\sigma(X)}|$ for all non-empty classes~$T_X$,
	we count $\prod_{X \subseteq S} |T_X|!$ automorphisms since, for this fixed~$\sigma$, each
	class~$T_X$ must map bijectively onto~$T_{\sigma(X)}$, which (as~$|T_X|=|T_{\sigma(X)}|$) can be done in~$|T_X|!$ ways regardless of~$\sigma$. As there are at most~$|T|$
	non-empty classes and querying~$|T_X|$ takes time~$O(|S|)$ with a
	suitable dictionary data structure (\ie a hash-map keyed by the neighbourhood-class index), we arrive at a running time of~$O
	(\|H\| + |S|! \cdot |T| \cdot |S|) = O(|S|^{|S|} \cdot |H|)$.
	
	If~$H$ is $d$-degenerate, we again can bound the number of neighbourhood classes by~$O(|S|^d)$ using Lemma~\ref{lemma:neighbourhood-classes} and thus get a running time of
	$O(\|H\| + |S|! |S|^{d+1}) = O(\|H\| + |S|^{|S| + d})$.
	\qed
\end{proof}

\begin{corollary}\label{cor:count-autos}
    For every pattern~$H$ the number of automorphisms $\aut(H)$ can be computed in time~$O(|S|^{|S|} \cdot |H|^4)$.
    If~$H$ is $d$-degenerate, the running time improves to~$O(d^2 |S|^{|S|+d} \cdot |H|^2)$.
\end{corollary}
\begin{proof}
    If we view $\aut(H)$ as acting on~$2^{V(H)}$ instead of~$V(H)$, then by the Orbit--Stabilizer Theorem:~$\aut(H) = |\orb_S(H)| \cdot \aut_S(H)$.
    Therefore we simply combine Lemma~\ref{lemma:S-orbit} and Lemma~\ref{lemma:count-S-autos}.
    %In both cases the running time of Lemma~\ref{lemma:S-orbit} dominates. 
    \qed
\end{proof}

\noindent
We will need to keep track of embeddings, automorphisms,
and orderings which poses a challenge for the presentation of our work and the patience of the reader. The following notion of an \emph{index representation}
encodes a specific automorphism of~$H$ under a specific ordering in which it might appear in a host graph. 

% In short, the index representation tells us the order of~$S$ in the host graph and how many joint $T$-neighbours each subset of~$S$ must have. In order to translate back and forth between~$H$ and~$G$, the vertices of the representation are normalised to the numbers~$[|S|]$ which additionally helps to encode their ordering.

\begin{definition}[Index representation]%
    \marginnote{Index representation, $(\H_S, J)$, $\rep$}%
   Let~$H$ be a pattern with sides $S$, $T$ and let~$\S$ be an ordering of~$S$. Let further $s \defeq |S|$. The \emph{index representation of~$H$ under $\S$} is a tuple~$\rep(H, \S) = (\H_S, \jointex)$
    where $\H_S$ is an ordered graph with vertices~$[s]$ and
    $\jointex\colon 2^{[s]} \to \N$ is a function such that
    \vspace*{-6pt}
    \begin{itemize}
        \item $\H_S$ is obtained from~$H[S]$ by relabelling vertices according to their index in $\S$ and applying the natural order on~$[s]$, and
        \item for all $I \subseteq [s]$ it holds that
        $
            \jointex(I) = \big| \{ v \in T \mid N_H(v) = \S[I] \} \big|
        $.
    \end{itemize}
    \vspace*{-6pt}
    We write $\rep(H) \defeq \{ \rep(H,\S) \mid \S \text{ an ordering of } S \}$ for the set of all index representations of~$H$.
\end{definition}

\noindent 
Put more simply, an index representation provides a uniform encoding of how the $S$-part of a pattern is ordered and accordingly how the neighbourhoods of $T$-vertices will be labelled. 
We will frequently need the index representation for the image of a given embedding of~$H$ into~$G$, which is why we introduce the following shorthand:

\begin{definition}[Index representation of embedding]%
    \marginnote{$\rep(\any, \G)$}%
    Let~$\lambda$ be either a strong or weak embedding of~$H$ into~$G$ and let~$\G$ be an ordering of~$G$. Let further~$\tilde H$ be the pattern in~$G$ defined by the image of~$\lambda$, meaning $\tilde H = \G[\lambda(V(H))]\setminus E(G[\lambda(T)])$ in the strong case and $\tilde H = \lambda(H)$ in the weak case.
    We define the notation
    $
        \rep(\lambda, \G) \defeq \rep(\tilde H, \tilde \S)
    $
    where $\tilde \S$ is the set~$\lambda(S)$ ordered by~$<_\G$.
\end{definition}

\noindent
We can illustrate this with an example to highlight why we require index representations. For $H = K_{2,3}$ with small side $S = \{s_1,s_2\}$ and large side $T = \{t_1,t_2,t_3\}$: 
given $S$ is independent, $\H_S$ is the edgeless graph on $\{0,1\}$; and given every $t_j$ has $N_H(t_j) = S$, the only non-zero value of $\jointex$ is $\jointex(\{0,1\}) = 3$.
So $\rep(H) = \{(\H_S,\jointex)\}$ which is a single representation, shared by every strong or weak embedding of $H$ into any host graph, 
since each $t_j$'s image must always have exactly the image of $S$ as its neighbourhood.

Contrast this with the pattern $H'$ that keeps $S = \{s_1,s_2\}$ but replaces $T$ with $\{t_1,t_2,t_3\}$ where $N_{H'}(t_1) = \{s_1\}$, $N_{H'}(t_2) = \{s_2\}$, and $N_{H'}(t_3) = \{s_1,s_2\}$: now $\jointex$ takes three distinct non-zero values, $\jointex(\{0\}) = \jointex(\{1\}) = \jointex(\{0,1\}) = 1$, since the three $T$-vertices attach to three different subsets of $S$.
The biclique-counting review at the start of this section only ever needs a single count of joint neighbours. Here, that would not suffice to distinguish these three cases; indexing $\jointex$ by subsets of $S$ is exactly what lets us track them separately.

Index representations are semi-canonical representations of patterns because we can use them to describe embeddings:
if $\phi$ is a strong embedding of $H$ into~$G$ and $\G$ is an ordering of~$G$, then there exists an ordering~$\S$ of~$S$ such that $\rep(H,\S) = \rep(\phi, \G)$. 
More specifically, we can use representations to distinguish certain embeddings:

\begin{lemma}\label{lemma:rep-diff}%
    Let~$\phi_1, \phi_2$ be embeddings (weak or strong) of~$H$ into~$G$ with~$\phi_1(S) = \phi_2(S)$ and let~$\G$ be an ordering of~$G$.
    If $\rep(\phi_1, \G) \neq \rep(\phi_2, \G)$ 
    then $\phi_1 \neq \phi_2$.
\end{lemma}
\begin{proof}
	Let $(\H^1_S, \jointex^1) = \rep(\phi_1, \G)$ and let
	$(\H^2_S, \jointex^2) = \rep(\phi_2, \G)$ be the two representations. If~$\H^1_S \neq \H^2_S$ then clearly~$\phi_1 \neq \phi_2$, so assume $\H^1_S = \H^2_S$. If the representations differ, there must exist an index set $I \subseteq [s]$ such that $\jointex^1[I] \neq \jointex^2[I]$. But then, by construction of the representations, the set $\tilde \S[I] \subseteq \tilde S$ has $\jointex^1[I]$ joint neighbours in~$\phi_1(T)$ and $\jointex^2[I]$ joint neighbours in~$\phi_2(T)$ and we conclude that $\phi_1 \neq \phi_2$.
	\qed
\end{proof}

\subsection{Counting fixed occurrences}\label{sec:counting-fixed}

We now describe the parts of our algorithm which count the number of strong/weak embeddings whose $S$-side coincides with some fixed set $\tilde S$ in the host graph~$G$. More concretely, these subroutines (Algorithm~\ref{alg:count-strong} and Algorithm~\ref{alg:count-weak}) are given a specific index representation and only count embeddings that conform to this embedding. We deal with finding the sets~$\tilde S$ in the second part of the algorithm.

\begin{algorithm}
    % \varrule[.8pt] 
    \DontPrintSemicolon
    \SetNoFillComment
    \SetKwProg{Fn}{Function}{}{}

    \KwInput{An ordered graph $\G$, an index representation $(\H_S, \jointex)$ of a pattern~$H$, a subset $\tilde S \subseteq V(G)$ of size~$|S|$, the data structure~$Q_{\tilde S}$ from Theorem~\ref{thm:data-structure}, and the automorphism group~$\aut_S(H)$.}
    \KwOutput{$\Count(H \embeds{\phi} G \mid \phi(S) = \tilde S ~\text{and}~\rep(\phi, \G) = (\H_S, \jointex))$.}
    
	\vspace*{-6pt}
    \varrule

    \Fn{CountStrong($\G$, $(\H_S, \jointex)$, $\tilde S, Q_{\tilde S}$, $\aut_S(H)$)}{
        $t = 0$\;
        Let $\tilde \S$ be $S$ ordered by $<_\G$\;
        \comment{1}{Check $S$-side of pattern}\;
        \If{$\H_S \neq \G[\tilde S]$}{
            \Return 0\; 
        }
        \comment{2}{Count number of choices for $T$-side of pattern}\;
        $k = 1$\;
        \For{$I \subseteq [s]$ with $\jointex(I) \neq 0$}{
            $k = k \cdot {Q_{\tilde S}[\tilde \S[I]] \choose \jointex(I) }$\;
        }
        \Return $\aut_S(H) \cdot k$ \tcp*{Normalised for ease of presentation}
    }
    % \varrule[.8pt] \\[1ex]
    \caption{\label{alg:count-strong}%
    Subroutine to count strong embeddings for a specific $S$-set.}
\end{algorithm}

\begin{replemma}{LemmaCountSStrong}%
% \begin{lemma}%
    \label{lemma:count-S-strong}%
    Let~$\G$ be an ordering of a graph~$G$ and let~$(\H_S, \jointex)$ be an index representation of a pattern~$H$. Let further~$\tilde S \subseteq V(G)$ and~$\tilde \S$ be the ordering of~$\tilde S$ under $<_\G$. Let~$Q_{\tilde S}$ be the data structure described in Theorem~\ref{thm:data-structure}.
    Given $\G$, $(\H_S, \jointex)$, $\tilde S$, and $Q_{\tilde S}$, the subroutine \emph{CountStrong} (Algorithm~\ref{alg:count-strong}) 
    returns
    \[
        \Count(H \smash{\embeds{\phi}} G \mid \phi(S) = \tilde S ~\text{and}~\rep(\phi, \G) = (\H_S, \jointex)),
    \]
    that is, the number of strong embeddings~$H \smash{\embeds{\phi}} G$ with $\phi(S) = \tilde S$ and whose representative under~$\G$
    is exactly~$(\H_S, \jointex)$.
    The running time of \emph{CountStrong} is~$O(|\supp(\jointex)| \cdot |S|)$
% \end{lemma}
\end{replemma}
\begin{proof}
    Let $\Phi \defeq \{ H \embeds{\phi} G \mid \phi(S) = \tilde S ~\text{and}~ \rep(\phi, \G) = (\H_S, \jointex) \}$, \ie the set of all embeddings that
    map $S$ onto~$\tilde S$ and have the given index representation.
    Let further~$\mathcal H \defeq \{ V(\phi(S \cup T)) \mid \phi \in \Phi \}$ contain all vertex subsets of~$V(G)$ that are images of the elements of~$\Phi$.

    First, we claim that $|\Phi| = \aut_S(H) \cdot |\mathcal H|$.
    This is the usual relationship between embeddings and substructures:
    Simply note that if we take an automorphism~$\lambda$ of~$H$ with
    $\lambda(S) = S$ and a strong embedding~$\phi \in \Phi$ with~$X = \phi(S \cup T)$,
    then~$\phi \circ \lambda$ is a different strong embedding which maps onto~$X$. In the other direction, we can take two distinct
    embeddings~$\phi,\phi' \in \Phi$ and the function~$\lambda$
    which satisfies~$\phi \circ \lambda = \phi'$ will be an automorphism of~$H$ with~$\lambda(S) = S$. 

    We claim that at the end of calling \emph{CountStrong}  with the above arguments, the variable~$k$
    is equal to~$|\mathcal H|$ and therefore \emph{CountStrong} returns
    the value $|\Phi|$.
    Note that this is equivalent to claiming that
    $
        |\mathcal H| = \prod_{\substack{I \subseteq [s]}} { Q_{\tilde S}[\tilde \S[I]] \choose \jointex(I)}
    $,
    as the right-hand side is precisely what the variable~$k$ contains at the end of the function.
    Recall that the data structure~$Q_{\tilde S}[X]$ counts the number of joint vertices in~$G$ whose neighbourhood in~$\tilde S$ is precisely~$X$. For~$X \subseteq \tilde S$, let us write
    $V_X \defeq \{ v \in V(G)\setminus S \mid N_G(v) \cap S = X \}$. With this notation, we see that~$Q_{\tilde S}[X]$ is exactly~$|V_X|$ and therefore the right-hand side is equal to
    \[
        \prod_{\substack{I \subseteq [s]}} { Q_{\tilde S}[\tilde \S[I]] \choose \jointex(I)} 
        = \prod_{\substack{X \subseteq \tilde S}} { Q_{\tilde S}[X] \choose \jointex(\ix_{\tilde \S}(X))}
        =  \prod_{\substack{X \subseteq \tilde S}} { |V_X| \choose \jointex(\ix_{\tilde \S}(X))}.
    \]
    To see that this right-hand side is exactly~$|\mathcal H|$, consider a vertex set~$\tilde H \subseteq V(G)$ that is the image of at least one~$\phi \in \Phi$. Since~$\phi(S) = \tilde S$, we know that~$\tilde S \subseteq \tilde H$ and hence~$\phi(T) = \tilde H \setminus \tilde S$. Moreover, since~$\rep(\phi, \G) = (\H_S, \jointex)$, for each~$X \subseteq \tilde S$ we have that~$\phi(T)$
    contains exactly~$J(\ix_{\tilde \S}(X))$ vertices of~$V_X$. Accordingly, the total number of sets in~$|\mathcal H|$ is exactly the product of choices for taking $J(\ix_{\tilde \S}(X))$ out of~$V_X$, which is exactly the right-hand side.

    The running time is straightforward assuming that the function~$\jointex$
    is stored using \eg a prefix trie and therefore allows us to iterate over its support in time~$O(|\supp(\jointex)|)$.  \qed
\end{proof}

\begin{algorithm}[H]
    % \varrule[.8pt]
    \DontPrintSemicolon
    \SetNoFillComment
    \SetKwProg{Fn}{Function}{}{}

    \KwInput{An ordered graph $\G$, an index representation $(\H_S, \jointex)$ of a pattern~$H$, a subset $\tilde S \subseteq V(G)$ of size~$|S|$, the data structure~$Q_{\tilde S}$ from Theorem~\ref{thm:data-structure}, and the automorphism group~$\aut_S(H)$.}
    \KwOutput{$\Count(H \embedsweak{\psi} G \mid \psi(S) = \tilde S ~\text{and}~\rep(\psi, \G) = (\H_S, \jointex))$}

    \vspace*{-6pt}
    \varrule

    \Fn{CountWeak${}_d$($\G$, $(\H_S, \jointex)$, $\tilde S$, $Q_{\tilde S}$, $\aut_S(H)$)}{
        Let $\tilde \S$ be $S$ ordered by $<_\G$\;
        \If{$\H_S \not \subseteq \G[\tilde S]$}{
            \Return 0\; 
        }        
        \comment{1}{Construct flow graph (s,t are implicit, see proof of Lemma~\ref{lemma:count-S-weak})}\;
        Let $A = (L, R, F)$ with $L, R, F = \emptyset$\;
        Initialise empty dictionaries $\supply, \demand$\;
        \For{$I \subseteq [s]$ with $Q_{\tilde S}[\tilde \S[I]] > 0$}{
            $L = L \cup \{ l_I \}$\;
            $\supply[l_I] = Q_{\tilde S}[\tilde \S[I]]$\;
        }
        \For{$I \subseteq [s]$ with $\jointex(I) > 0$}{
            $R = R \cup \{ r_I \}$\;
            $\demand[r_I] = \jointex(I)$\;
        }      
        \For(\tcp*[f]{See Lemma~\ref{lemma:count-S-weak} for details }){$l_I, r_{I'} \in L \times R$ with $I' \subseteq I$ }{ 
            $F = F \cup (l_I, r_{I'})$
        }                    
        \comment{2}{Enumerate embeddings}\;
        $p = \max_{r_I \in R}\demand[r_I]$, \,
        $M = \sum_{r_I \in R}\demand[r_I]$, \,
        $t = 0$\; 
        \For{$f \in [p]^F$}{
            \tcp{We treat edges in~$F$ as having infinite capacity}
            \If{$f$ is not a flow on $A$ or has value $< M$ }{
                \Continue
            }
            \comment{3}{Compute number of embeddings represented by $f$}\;                
            $k = 1$\;
            \For{$l_I \in L$}{
                $k = k \cdot \frac{\supply[l]!}{(\supply[l] - f(sl))!}$
            }
            $t = t + k$\;
        }
        \Return $\aut_S(H) \cdot t$ \tcp*{Normalised for ease of presentation}
    }
    % \varrule[.8pt] \\[1ex]
    \caption{\label{alg:count-weak}%
    Subroutine to count weak embeddings for a specific $S$-set.}
    % based on an algorithm by Curticapean and Marx~\cite{curticapeanMarxCount2014}}
\end{algorithm}

\noindent
The proof for counting weak patterns is similar, but more involved since we have more choices of which edges to include.

\begin{lemma}\label{lemma:count-S-weak}%
    Let~$\G$ be a $d$-degenerate ordering of~$G$ and let~$(\H_S, \jointex)$ be an index representation of a $d$-degenerate pattern~$H$. Let further~$\tilde S \subseteq V(G)$ and~$\tilde \S$ the ordering of~$\tilde S$ under $<_\G$. Let~$Q_{\tilde S}$ be the data structure described in Theorem~\ref{thm:data-structure}.
    Given $\G$, $(\H_S, \jointex)$, $\tilde S$, and $Q_{\tilde S}$, the subroutine \emph{CountWeak${}_d$} (Algorithm~\ref{alg:count-weak}) 
    returns 
    \[
        \Count(H \smash{\embedsweak{\psi}} G \mid \psi(S) = \tilde S ~\text{and}~\rep(\psi, \G) = (\H_S, \jointex)),
    \]
    that is, the number of weak embeddings~$H \smash{\embedsweak{\psi}} G$ with $\psi(S) = \tilde S$ and whose representative under~$\G$
    is exactly~$(\H_S, \jointex)$.
    The running time of \emph{CountWeak${}_d$} is~$O(d2^d |S|^{d+2} \cdot |H|^{2^d |S|^{d+1}}) $    
.
\end{lemma}
\begin{proof}
    Define the set $\Psi \defeq \{ H \embedsweak{\psi} G \mid \psi(S) = \tilde S ~\text{and}~ \rep(\psi, \G) = (\H_S, \jointex) \}$ to contain all weak embeddings that
    map $S$ onto~$\tilde S$ and have the given index representation.
    Let further~$\mathcal H \defeq \{ \psi(H) \mid \psi \in \Psi \}$ contain all subgraphs of~$G$ that are images of~$\Psi$.
    As in the case for strong embeddings, it is easy to see that
    $|\Psi| = \aut_S(H) \cdot |\mathcal H|$.

    We claim that at the end of calling \emph{CountWeak${}_d$}  with the above arguments, the variable~$t$ is equal to~$|\mathcal H|$ and thus the 
    function returns~$|\Psi|$.

    Let us again use the notation~$V_X \defeq \{ v \in V(G)\setminus \tilde S \mid N(v) \cap \tilde S = X \}$ for~$X \subseteq \tilde S$
    and introduce the notation~$V_{X \subseteq{}} \defeq \{ v \in V(G)\setminus \tilde S \mid N(v) \cap \tilde S \supseteq X \}$.
    In order to construct a graph~$\tilde H \in \mathcal H$, which in particular must have the index representation~$(\H_S, \jointex)$, we need to add for every set~$X \subseteq \tilde S$ a total
    of $\jointex(\ix_{\tilde \S}(X))$ vertices from~$V_{X\subseteq{}}$ to~$\tilde H$. To enumerate all possible ways in which we can distribute vertices from~$\{V_{X\subseteq{}}\}_{X \subseteq \tilde S}$ onto $T$-vertices, \emph{CountWeak${}_d$} proceeds as follows
    (based on an algorithm by Curticapean and Marx~\cite{curticapeanMarxCount2014}).

    First, we construct a bipartite flow graph~$A = (s,t,L,R,F)$ with source $s$ and sink~$t$. Edges will only go from~$s$
    to~$L$, from~$L$ to~$R$, and from~$R$ to~$t$.  The set $L$ contains a vertex~$l_I$ for each index set~$I \subseteq [s]$ where~$Q_{\tilde S}[\tilde \S[I]] = |V_{\S[I]}|$ is larger than zero and the set~$R$ contains vertices~$r_I$ for all index sets $I \subseteq [s]$ where~$\jointex(I) > 0$.  The capacities of
    edges from~$s$ to~$L$ are stored in a dictionary~$\supply$, the capacities of edges from~$L$ to~$R$ are all infinite, and the capacities of edges from~$R$ to~$t$ are stored in a dictionary~$\demand$. All edges are then added to $F$. 

    Next the algorithm adds each edge~$l_{I} r_{I'}$ to~$F$ where $I' \subseteq I$ with infinite capacity. This models that the set~$\S[I]$ can be used to construct vertices with a neighbourhood~$\S[I'] \subseteq \S[I]$.
    For reasons of space, the listed code simply enumerates all these edges, but the procedure is a bit more involved to arrive at a better running time:
    \begin{claim}
        The number of all pairs~$l_{I}, r_{I'}$ with $I' \subseteq I$
        is at most~$O(2^d |S|^{d+1})$
        and we can enumerate them in time~$O(d2^d |S|^{d+2})$.
    \end{claim}
    \begin{proof}
        Since~$H$ is $d$-degenerate, we know by Lemma~\ref{lemma:neighbourhood-classes} that~$T$ contains at most~$(d+1)|S|$ vertices of degree~$> d$ and that the remaining vertices can be
        partitioned into $3|S|^d$ neighbourhood classes. We define
        $R_{> d} \defeq \{ r_I \in R \mid |I| > d \}$ and
        $R_{\leq d} \defeq R \setminus R_{> d}$. By the above,
        $|R_{> d}| \leq (d+1)|S|$ and $|R_{\leq d}| \leq 3|S|^d$.

        Similarly, because~$G$ is~$d$-degenerate, so is $G[\tilde S \cup N(\tilde S)]$, therefore the same argument applies to the
        sets~$L_{> d} \defeq \{ l_I \in L \mid |I| > d\}$
        and~$L_{\leq d} \defeq L \setminus L_{> d}$; namely that
        $|L_{> d}| \leq (d+1)|\tilde S| = (d+1)|S|$ and $|L_{\leq d}| \leq 3|\tilde S|^d = 3 |S|^d$.

        We now enumerate all pairs~$r_I, l_{I'}$ with~$I' \subseteq I$ as follows: we test each vertex in~$R_{> d}$ against each vertex
        in~$L$ in time~$O(|R_{>d}| \cdot |L| \cdot |S|) = O(|S|^{d+2})$,
        where we incur a factor of~$O(|S|)$ for the lookup operation, and output those~$O(|S|^{d+1})$ pairs for which the inclusion relationship holds.

        Next, for each vertex~$r_I \in R_{\leq d}$ we test for each
        subset~$I' \subseteq I$ whether~$l_{I'}$ exists and if so we
        output the pair~$r_I, l_{I'}$, this takes time
        $O(d2^d \cdot |R_{\leq d}|) = O(d2^d |S|^d)$ where we incur
        a factor of~$O(d)$ for looking up vertices by their index set in a suitable data structure. 
    \end{proof}
    \noindent
    To summarise the time needed to construct~$A$:
    constructing both vertex sets takes time~$O(\supp(\jointex) + \supp(Q_{\tilde S}))$ which as observed above is bounded by
    $O(|S|^{d+1})$ and adding the edges between the sets costs~$O(2^d |S|^{d+1})$ which also subsumes the whole construction.

    Now note that any flow on~$A$ that reaches full capacity on all 
    edges incident to~$t$ models one way of distributing vertices from
    the sets~$\{V_{X\subseteq{}}\}_{X \subseteq \tilde S}$ onto $T$-vertices of~$H$ (via the index representation). We can easily test
    this either by inspecting the flow on these edges or checking, as the algorithm does, whether the flow value is equal to
    $\sum_{I \subseteq [s]} \jointex(I)$. Let~$f$ be such a flow,
    then we can compute the number of embeddings that distribute vertices
    according to~$f$ as
    \[
        \prod_{l \in L}
        \frac{\mathsf{supply}[l]!}
        {(\mathsf{supply}[l] - f(sl))!}.
    \]
    To understand why, for each $l$ in $L$ we must decide how to distribute vertices onto the $T$-vertices. In each case, if we assign the first
    $\demand[r_{I_1}]$ vertices to map onto the $T$-vertices with neighbours~$I_1$ in the index representation, then the next $\demand[r_{I_2}]$ vertices to map onto those with neighbours~$I_2$, and so on, we want to consider all possible permutations of the $\supply[l]$ vertices, which is $\supply[l]!$. However, since the actual flow into $l$, $f(sl)$, can be less than the total supply to $l$, enumerating all these permutations over-counts, in each case by a factor of $(\supply[l] - f(sl))!$, which we must correct for.
    All in all,
    we have  $\frac{\supply[l]!} {(\supply[l] - f(sl))!}$ many choices, taking the product over all vertices in~$L$ then gives us the above formula, which is exactly what the algorithm accumulates in the variable~$k$ in part~\commentmark{3}. 

    Since the algorithm brute-forces all candidate functions~$f \in [p]^{F'}$ where $p = \max_{r \in R} \demand[r]$ and~$F' = F \cap (L \times R)$ and filters those that are not flows or do not have maximum flow value, we conclude that the variable~$t$ at the end of the function indeed contains the value~$|\mathcal H|$.

    The running time of this latter operation is dominated by enumerating and checking all functions~$f$ which takes time
    $
        O( p^{|F|} \cdot |F| ) = O(|H|^{2^d |S|^{d+1}} \cdot 2^d |S|^{d+1}) 
    $
    where we used the trivial bound~$p = \max_{r \in R} \demand[r] \leq |H|$. This running time, as all other running times of earlier parts, are subsumed by the claimed running time. \qed
\end{proof}

\subsection{Counting all occurrences}\label{sec:algorithm}

Now that we have algorithms in place to count occurrence of patterns for
specific sets~$\tilde S$ in the host graph, we want to find potential
candidate set in time better than~$O(|G|^{|S|})$. This is
clearly not possible for all patterns, as shown by the various existing lower
bounds, and the following definitions capture the types of patterns our
algorithm can count. 

First we provide a definition of the left-cover number as this allows us to bound
the number of vertices we must guess in the host graph to then locate a candidate
set $\tilde S \subseteq V(G)$ corresponding to a pattern's small side $S$.

\begin{definition}[Left-cover number]
    The left-cover number of a vertex set~$S$ in an ordered graph~$\H$
    is the size of the minimum vertex set~$C$ (not necessarily disjoint
    from~$S$) such that~$S \subseteq N^-_\H[C]$. We denote this number
    by~$\gamma^-_S(\H)$.
\end{definition}

\noindent
Now we provide the following definition for a general subset $S' \subseteq S$ of a
pattern $H$. We provide it as such, rather than for $S$ itself, as in
Section~\ref{sec:beyond-locatability} we reason about the locatability of subsets of $S$.
In the case where we are only interested in whole-pattern locatability, we can simply read~$S' = S$.

\begin{definition}[$(c,d)$-locatable]
    A set~$S' \subseteq S$ in a pattern~$H$ is \emph{$(c,d)$-locatable} if for every $d$-degenerate ordering $\H$ of~$H$
    it holds that $\gamma^-_{S'}(\H) \leq c$.
    We say that a pattern is \emph{$(c,d)$-locatable} if $S$ is
    $(c,d)$-locatable.
\end{definition}

\noindent
Note that by this definition, all patterns with degeneracy larger than~$d$ are~$(\any, d)$-locatable. This is intentional: if we ask how often a pattern~$H$ is contained in some $d$-degenerate graph~$G$ and $H$'s degeneracy is larger than~$d$, we can immediately answer this question.
As a consequence, if a pattern
is~$(c,d)$-locatable, then it is $(c',d)$-locatable for every~$c' \geq c$
and $(c,d')$-locatable for every~$d' \leq d$. We note that every pattern is~$(|S|,\infty)$-locatable. Let us next show how the left-cover number can be computed:

\begin{lemma}\label{lemma:left-cover-compute}
    For every pattern~$H$ and~$d \in \N$, let~$c\in\N$ be the minimum value such that~$H$ is~$(c,d)$-locatable. We can compute~$c$ in time~$O\big(d|H| + (cd)^{O((cd)^d)}\big)$.
\end{lemma}
\begin{proof}
    If the pattern~$H$ is not $d$-degenerate then there is nothing to be done,
    we can check this in~$O(\|H\|) = O(d|H|)$ time.
    Otherwise,  by Lemma~\ref{lemma:neighbourhood-classes} we can partition~$T$ into at most~$O(|S|^d)$ neighbourhood classes~$\{T_X\}_{X \subseteq S}$ and
    again this can be done in~$O(|H|+\|H\|) = O(d|H|)$ time.
    Let~$T'$ contain one representative vertex for each class.

    For some given ordering~$\S$ of~$S$, note that the value of the left-cover depends on the location of~$T$-vertices relative to~$\S$, but not on the position of~$T$-vertices relative to each other. Moreover, every minimum left-cover contains at most one vertex of each class~$T_X$.

    Accordingly, in order to determine~$c$ for some fixed ordering~$\S$ 
    of~$S$, it is enough to consider all ways in which~$T'$ can be placed in the~$|S|+1$ places between vertices in~$\S$. However, we must further determine whether this placement is compatible with
    some $d$-degenerate ordering and to do so efficiently we need to refine this idea further.

    Let us call a neighbourhood class~$T_X$ \emph{big} if $|T_X| > d$
    and \emph{small} otherwise. Observe that if~$T_X$ is big, then
    in any $d$-degenerate ordering~$\H$ of~$H$ it is true that
    $\max_{\H} T_X > \max_{\H} X$ since otherwise the vertex~$\max_{\H} X$ would have more than~$d$ left neighbours. In that case, we do not need to know
    where the vertex~$t = \max_{\H} T_X$ is actually placed since already~$N_{\H}^-(t) = N_H(t)$.

    We therefore proceed as follows to determine~$c$. Let~$T_s$ contain
    all vertices in~$T'$ that belong to small classes and let~$T_b$ contain a single representative for each large class.

    We now enumerate each possible way of placing the representatives~$T_s$ in the $|S|+1$ positions between vertices in~$\S$. The number of such possible placements is 
    \[
        (|S|+1)^{|T_s|}
        = O\big((|S|+1)^{4|S|^d}\big)
        = (cd)^{O((cd)^d)}
    \]
    where we used the bound from Lemma~\ref{lemma:neighbourhood-classes} to bound~$|T_s|$ and that~$|S| \leq cd$.
    For each placement, we check whether it corresponds to a $d$-degenerate ordering by checking whether each vertex in~$S \cup T_s$
    has at most~$d$ left neighbours, if not we discard this placement.
    Then, to determine~$c$ for the current placement, we check all subsets of~$S \cup T_s \cup T_b$ of size~$\leq c$ in time~$O(cd \cdot |S \cup T_s \cup T_b|^c) = O\big(cd 4^c |S|^{cd}\big) = O\big(4^c (cd)^{cd+1}\big)$ and note the minimum value for which the subset forms a left-cover of~$S$. The total running time of this procedure is then
    \[
        O\Big(d|H| + |S|!\big((cd)^{O((cd)^d)}+ 4^c (cd)^{cd+1}\big)\Big) 
        = O\big(d|H| + (cd)^{O((cd)^d)}\big). \tag*{\qed}
    \]
\end{proof}

\noindent
The proof of Lemma~\ref{lemma:left-cover-compute} can further be used to compute all relevant representations of a pattern, meaning those that can occur in a $d$-degenerate host graph.

\begin{corollary}\label{corr:degenerate-reps-compute}
    For a pattern~$H$ and an integer~$d \in \N$, with~$c$ as in Lemma~\ref{lemma:left-cover-compute}, we can compute the set~$\{ \rep(H, \H[S]) \mid \H~\text{is a $d$-degenerate ordering}\}$ in time $O(d|H| + (cd)^{O((cd)^d)})$.
\end{corollary}
\begin{proof}
    We proceed as in the proof of Lemma~\ref{lemma:left-cover-compute}. For every ordering~$\S$ of~$S$ and any placement of~$T_s$ among~$
    \S$ we can check whether this placement corresponds to a $d$-degenerate ordering of~$H$. If so, we add the representative
    $\rep(H, \S)$ to the set. \qed
\end{proof}

\begin{algorithm}[t]
    % \varrule[.8pt] 
    \DontPrintSemicolon
    \SetNoFillComment

    \KwInput{A $d$-degenerate graph~$G$, a pattern~$H$}
    \KwOutput{$\Count(H \embeds{} G)$ or $\Count(H \embedsweak{} G)$}

	\vspace*{-6pt}
    \varrule[.4pt]

    Compute $d$-degenerate ordering $\G$ of $G$\;
    Initialise and fill subset dictionary~$R$\tcp*{Lemma~\ref{lemma:R}}
    Compute $\aut_S(H)$ \tcp*{Needed in CountStrong/CountWeak}    
    \tcp{Compute index reps.\ of $d$-degenerate orderings and left-cover number}
    Compute minimum $c$ for which~$H$ is~$(c,d)$-locatable \tcp*{Lemma~\ref{lemma:left-cover-compute}}
    Compute the set $\mathcal H = \{ \rep(H, \H[S]) \mid \H~\text{is $d$-degenerate}\}$ \tcp*{Corollary~\ref{corr:degenerate-reps-compute}}
    \tcp{Count frequencies of~$H$ by index representation}
    $k = 0$\;    
    \For{$L \subseteq V(G)$ with $|L| \leq c$}{
        \For{$\tilde S \subseteq N^-_\G[L]$}{
            Compute $Q_{\tilde S}$ from~$R$ \tcp*{Theorem~\ref{thm:data-structure}}
            \For{every $(\H_S,\jointex) \in \mathcal H$}{
                $k = k + \emph{CountStrong/CountWeak${}_d$}(\G, (\H_S, \jointex), \tilde S, Q_{\tilde S}, \aut_S(H))$
            }
        }
    }
    \Return $k$\;
    % \varrule[.8pt] \\[1ex]
    \caption{\label{alg:count}%
        Algorithm for counting $(c,d)$-locatable patterns in
        $d$-degenerate graphs.
    }
\end{algorithm}

\begin{theorem}\label{thm:count-strong}
    There exists an algorithm that, given
    a $d$-degenerate graph~$G$ and a $(c,d)$-locatable pattern~$H$ as input, computes $\Count(H \embeds{} G)$ in time\footnote{For ease of presentation we assume for the running time that~$d \geq 2$. This is
    justified by the fact that graphs of degeneracy one are trees.}
    $
        O\big(
             (cd)^{O((cd)^d)}
            + (cd)^{2cd+d} |G|^c \big).
    $
\end{theorem}
\begin{proof}
    The algorithm is listed as Algorithm~\ref{alg:count}. We first compute a $d$-degeneracy ordering~$\G$ of~$G$ in time~$O(\|G\|)$,  then~$\aut_S(H)$ in time~$O(|S|^{|S|+d} + \|H\|)$ (Lemma~\ref{lemma:count-S-autos}), the minimum~$c$
    for which $H$ is~$(c,d)$-locatable as well as the set~$\mathcal H$
    in time~$O(d|H| + (cd)^{O((cd)^d)})$ (Lemma~\ref{lemma:left-cover-compute} and Corollary~\ref{corr:degenerate-reps-compute}).

    The algorithm now uses that~$H$ is~$(c,d)$-locatable:
    for each strong embedding of~$\smash{H \smash{\embeds{\phi}} G}$ we can locate the set~$\phi(S)$ in the left-neighbourhood of some set~$L \subseteq V(G)$ of size at most~$c$.
    Therefore, we enumerate all such potential sets~$L$ in time
    $O(|G|^c)$ and for each~$L$ check all~$O((cd)^{|S|})$ candidate sets~$\tilde S \subseteq N^-[L]$. For each set~$\tilde S$, we invoke
    \emph{CountStrong} (Lemma~\ref{lemma:count-S-strong}) a total number of~$|\mathcal H| \leq |S|!$ times, resulting in a running time of
    \begin{align*}
        &\phantom{{}={}} O\big(|G|^c  (cd)^{|S|} |S|! |S|^{d+1} \big) 
        = O\big( (cd)^{cd}  (cd)! (cd)^{d+1} |G|^c \big) 
        % &= O\big( (cd)^{2cd} \cdot e^{-cd}  \cdot (cd)^{d+1} |G|^c \big) 
        = O\big( (cd)^{2cd+d} |G|^c \big)
    \end{align*}
    where we used that~$|S| \leq cd$ and that~$\supp(\jointex) \leq 4|S|^d$ by Lemma~\ref{lemma:neighbourhood-classes} for each representative~$(\H_S, \jointex) \in \mathcal H$.
    The total running time is
    \begin{align*}
        &\phantom{{}={}} O(\|G\|) 
            + O(|S|^{|S|+d} + \|H\|) 
            + O(d|H| + (cd)^{O((cd)^d)}) 
            + O\big( (cd)^{2cd+d} |G|^c \big) \\
          % &= O\big(d|G|
          %   + d|S|^{|S|+d} + d|H|
          %   + d|H| + (cd)^{O(cd^2)} 
          %   + (cd)^{2cd+d} |G|^c \big) \\
          &= O\big(
             (cd)^{O((cd)^d)}
            + (cd)^{2cd+d} |G|^c \big)  \tag*{\qed}
    \end{align*}
\end{proof}

\begin{theorem}\label{thm:count-weak}
    There exists an algorithm that, given
    a $d$-degenerate graph~$G$ and a $(c,d)$-locatable pattern~$H$ as input, computes $\Count(H \embedsweak{} G)$ in time
    $
        O\big((cd)^{O((cd)^d)}
            + cd^2 f(c,d) \cdot |H|^{f(c,d)} \cdot  |G|^c \big)
    $
    where~$f(c,d) = 2^d (cd)^{d+1}$.
\end{theorem}
\begin{proof}
    The proof works exactly as for Theorem~\ref{thm:count-strong}, the only difference being the running time incurred by invoking \textit{CountWeak}${}_d$ instead of \textit{CountStrong} and then using
    that~$|S| \leq cd$ to arrive at the above running time. \qed
    % \[
    %     O\big( (cd)^{O(cd^2)} +  d2^d |S|^{d+2} \cdot |H|^{2^d |S|^{d+1}} \cdot  |G|^c\big)
    % \]
\end{proof}

\subsection{The structure of $(1,d)$-locatable patterns}\label{sec:structure}

While the general structure of these patterns appears to be rich,
the structure of~$(1,d)$-locatable patterns---which are probably the only patterns interesting in practical applications---can be characterised neatly:

\begin{lemma}\label{lemma:1-locatable-structure}
    A pattern~$H$ is~$(1,d)$-locatable and~$d$-degenerate
    iff one of the following holds: 1. $H[S]$ is a clique, or 2.
    $T$ contains~$d+1 - \delta(H[S])$ vertices whose neighbourhood is~$S$.
\end{lemma}
\begin{proof}
    Clearly a pattern~$H$ where~$H[S]$ is a clique (where we include the case~$|S| = 1$) is $(1,d)$-locatable,
    let us therefore focus on the second: let~$A \subseteq T$ contain all vertices in~$T$ whose neighbourhood is~$S$ with~$|A| \geq d+1 - \delta(H[S])$. We claim that in any degeneracy ordering~$\H$ of~$H$, at least one vertex of~$A$ lies to the right of~$S$, \ie $\max A >_\H \max S$, which of course makes~$S$ $(1,d)$-locatable in this ordering. 
    Assume towards a contradiction that in~$\H$ we have~$\max A <_\H \max S$. Let~$x \defeq \max_\H S$, then $A \subseteq N^-(x)$. However,
    $
        |N^-(x)| \geq \delta(H[S]) + d + 1 - \delta(H[S]) = d+1
    $,
    contradicting that~$\H$ is a $d$-degenerate ordering.

    To prove the reverse, assume~$H$ is~$(1,d)$-locatable and~$d$-degenerate. We assume that~$H[S]$ is not a clique, so let again~$A \subseteq T$ contain all vertices whose neighbourhood is all of~$S$.
    We first note that~$A$ cannot be empty since otherwise no single vertex
    can left-cover~$S$. For the same reason, $|S| \leq d$.
    Assume towards a contradiction that~$|A| \leq d - \delta(H[S])$ and
    let~$x \in S$ be a vertex of minimum degree in~$H[S]$. We construct the
    following ordering~$\hat H$ of~$H$ and claim that it is $d$-degenerate but $S$ cannot be left-covered by a single vertex:
    \[
        S\setminus \{s\} ~\big|~ A ~\big|~ s ~\big|~ T \setminus A,
    \]
    where the ordering inside each block is arbitrary. First let us verify that it is indeed $d$-degenerate. The vertices in the first block have at most~$|S\setminus \{s\}| - 1 = d-2$ left neighbours. The vertices in~$A$ all have the same neighbourhood~$S$, so each vertex in the second block has all of~$|S\setminus\{s\}| = d-1$ left neighbours.
    Since~$s$ has minimum degree in~$H[S]$ it has 
    $
        \delta(H[S]) + |A| = d
    $
    left neighbours. Finally, all vertices in~$T \setminus A$ see some subset of~$S$ and thus all vertices in the fourth block have less than~$d$ left neighbours. Observe now that~$S$ in this ordering cannot be left-covered by a single vertex: $s$ has less than $|S|-1$ neighbours (otherwise~$H[S]$ is a clique) and so has every vertex in $T\setminus A$ by construction of~$A$. We conclude that this pattern is indeed not~$(1,d)$-locatable. \qed
\end{proof}

\noindent
The above characterisation of~$(1,d)$-locatable patterns immediately poses the question whether a similar characterisation holds for~$(1,d)$-locatable \emph{subsets}~$S' \subseteq S$. This turns out not to be true; we defer this discussion to the full version\footnote{The full version can be found here: http://arxiv.org/abs/2511.20385.} which also includes some further notes on~$(c,d)$-locatable patterns as well as remarks on locatability beyond distance one.

\section{Lower bounds}\label{sec:lower-bound}

% FR: Still discussing this part
While the general structure of~$(c,d)$-locatable patterns is complicated,
from a practical perspective we are mostly interested in $(1,d)$-locatable patterns, since these are the \emph{only} large patterns that can be counted in linear time  in degenerate graphs\footnote{Excluding graphs on five or less vertices for which this is still possible.}, assuming the TDC (\cite{beraCountingSix2020}). We complement this lower bound by showing that patterns which are \emph{barely} not~$(1,d)$-locatable---adding one more vertex with neighbourhood~$S$ makes them $(1,d)$-locatable---are $\#W[1]$-hard to count.

% Every pattern $H$ whose large side contains at least $d+1$ vertices whose degree is the size of the small side is $(2,d)$-locatable. Let use denote this set of vertices by $B$.
% This follwos from the fact that in any degenrate ordering of such pattern, one of the vertices of $B$ must be to the right of all of the vertices of the small side of $H$. 

% A natural question that arrises, is whether $d+1$ is the best possible.
% The following theorem shows, that in this simple criteria this is indeed the case.

\begin{theorem}
    For every~$d \geq 5$ there exists a pattern~$H$ with sides~$S,T$ such that 
    the problem of counting the number of strong or weak containments of~$H$ in
    a $d$-degenerate graph is
    $\#W[1]$-hard when parameterised by~$|S| + d$.
    Moreover, the pattern has exactly~$d$ vertices of degree~$|S|$ in~$T$
    and~$S$ is an independent set.
\end{theorem}
\begin{proof}
We prove the theorem by reducing counting cliques in general graphs (the canonical $\# W[1]$-hard problem~\cite{flumCountingHarndess2004}) to counting the number of weak containments of a specific pattern. The proof for strong containments works the same.

Given an instance~$(G',k)$ of the $k$-clique counting problem with~$k \geq 5$, we construct a graph $G$ and pattern $H$ such that the number of $k$-cliques in~$G'$ is equal to~$\Count(H \embedsweak{} G) / \aut(H)$.

To construct~$G$ we first subdivide each edge~$uv \in G'$ once. Let~$A$ be the set of all subdivision vertices created this way and let~$B = V(G')$. Note that the subdivided graph is
$2$-degenerate and bipartite with sides~$A, B$.

Next, add a set~$D_A$ of~$d-2$ new vertices with neighbourhood~$A$ each
and a set~$D_B$ of~$d$ new vertices with neighbourhood $B \cup D_A$ each.
The resulting graph is~$d$-degenerate, as witnessed by the following deletion sequence: first we delete the vertices in~$A$ which all have a degree of~$2+|D_A| = d$. In the remaining graph, $B$ is now only connected to~$D_B$ and therefore the vertices in~$B$ all have a degree of~$|D_B| = d$. Then we can delete~$D_A \cup D_B$ in any order since all vertices in the remaining graph have degree~$d$ or~$d-2$.

We generate our pattern $H$ by starting with a $k$-clique and applying the exact same process we used to construct $G$. Let us call the resulting sets~$\hat A$, $\hat B$, $D_{\hat A}$, and~$D_{\hat B}$ to distinguish them from the sets of~$G$. The sides of~$H$ are~$S = \hat B \cup D_{\hat A}$ and $T = \hat A \cup D_{\hat B}$, note that~$|S| = k + d - 2$ and~$|T| = \frac{1}{2}k(k-1) + d$, so in particular $|S| < |T|$. Note further that only the~$d$ vertices in~$D_{\hat B} \subseteq T$ have all of~$S$ as their neighbourhood.

\begin{claim}
    Let~$H \smash{\embedsweak{\psi}} G$. Then~$\psi(D_{\hat A}) = D_A$ and~$\psi(D_{\hat B}) = D_B$.
\end{claim}
\begin{proof}
	Let $X = D_A \cup B$, $Y = D_B \cup A$, $\hat X = D_{\hat A} \cup \hat B$ and~$\hat Y = D_{\hat B} \cup \hat A$. We first show that
	either $\psi(\hat X) = X$ and~$\psi(\hat Y) = Y$, or $\psi(\hat X) = Y$ and~$\psi(\hat Y) = X$. To see that, note that both $\hat X$ and~$\hat Y$ have (at least one) vertex~$\hat a \in V(H)$ such that $N_H(\hat a)$ contains all of $\hat X$ or $\hat Y$, respectively. Since~$X,Y$ is a bipartition of~$G$, we immediately have that each of~$\hat X$ and~$\hat Y$ can only be mapped entirely into~$X$ or entirely into~$Y$, as otherwise the vertex~$\hat a$ could not be mapped. Since there are edges between $\hat X$ and~$\hat Y$ in~$H$, we further get that~$\psi(\hat X) = X$ implies that~$\psi(\hat Y) = Y$ and $\psi(\hat X) = Y$ implies~$\psi(\hat Y) = X$.

	Consider first the case that~$\psi(\hat X) = Y$, \ie $\psi(D_{\hat A} \cup \hat B) = D_B \cup A$ and~$\psi(D_{\hat B} \cup \hat A) = D_A \cup B$. Since~$D_{\hat B}$ has size~$d$ and~$D_A$ has size~$d-2$, at least two vertices~$\hat x,\hat y$ of~$D_{\hat B}$ must be mapped onto~$B$. Note then that~$|\psi(D_{\hat A}) \cap A| \leq 1$, as otherwise
	we would have two vertices from~$D_{\hat A}$ mapped onto~$A$ and
	two vertices of~$D_{\hat B}$ mapped onto~$B$, but these would form a $K_{2,2}$ between~$A$ and~$B$ which cannot exist by construction of~$G$. This means that at least~$d-3$ vertices of~$D_{\hat A}$ are mapped onto~$D_B$, and therefore that at least~$|\hat B| - 3 = k - 3$ of~$\hat B$ must be mapped onto~$A$. Note that all such vertices must be connected to both~$\hat x$ and~$\hat y$ (since~$\hat x, \hat y \in D_{\hat B}$), but then for~$k \geq 5$ we arrive at a contradiction because we would need a $K_{2,2}$ between~$A$ and~$B$ for the embedding to work.

	Consider now the case that~$\psi(\hat X) = X$, \ie $\psi(D_{\hat A} \cup \hat B) = D_A \cup B$ and~$\psi(D_{\hat B} \cup \hat A) = D_B \cup A$. We now show that the claim must hold, so assume towards a contradiction that~$\psi(D_{\hat B}) \cap A$ is non-empty and let
	$x \in \psi(D_{\hat B}) \cap A$. Then~$x$ must be connected to all of~$\psi(D_{\hat A} \cup \hat B)$, but note that~$x$ has only two neighbours in~$B$. Accordingly, at most two vertices of~$D_{\hat A} \cup \hat B$ can be mapped onto $B$, leaving at most~$2 + |D_A| = d$ vertices as potential targets for mapping~$D_{\hat A} \cup B$.
	Since $|D_{\hat A} \cup \hat B| = d-2 + k$, for~$k \geq 3$ no such mapping can exist. As we assumed that~$k \geq 5$, we conclude that the claim holds for~$\psi$.
\end{proof}

\noindent
As a consequence, any weak embedding will also satisfy~$\psi(\hat A) \subseteq A$ and~$\psi(\hat B) \subseteq B$. In fact~$\psi(\hat B)$ already completely determines~$\psi(\hat A)$ since for~$\hat b_1, \hat b_2 \in \hat B$ the vertex~$\hat a \in \hat A$ with~$N_H(\hat a) \setminus D_{\hat A}= \{\hat b_1, \hat b_2\}$ necessarily must be mapped onto~$a \in A$ with~$N_G(a)\setminus D_A = \{\psi(\hat b_1), \psi(\hat b_2)\}$.

To now relate the number of $k$-cliques in~$G'$ to the number of weak embeddings of~$H$ into~$G$, we consider the following equivalence classes of embeddings:
For~$X \subseteq B$ let~$\Psi_X = \{ H \smash{\embedsweak{\psi}} G \mid \psi(\hat B) = X \}$ contain all weak embeddings of~$H$ that map~$\hat B$ onto~$X$. 

\begin{claim}
    For all~$X \subseteq B$, either $|\Psi_X| = 0$ or $|\Psi_X| = k! \cdot (d-2)! \cdot d! = \aut(H)$.
\end{claim}
\begin{proof}
    Since~$A$ consists precisely of the subdivision vertices of edges of~$G'$, and~$\psi(\hat A)$ is completely determined by~$\psi(\hat B) = X$ (as noted above), a~$\psi \in \Psi_X$ exists iff every pair of~$\hat B$ has a witnessing edge in~$G'$, \ie iff~$X$ induces a $k$-clique in~$G'$ (recall~$\hat B$ induces a $k$-clique in~$H$).

    In that case, simply note that once the embedding of~$\hat B$ onto~$X$ is fixed, we have~$(d-2)! \cdot d!$ ways in which to map~$D_{\hat A}$ onto~$D_A$ and~$D_{\hat B}$ onto~$D_B$. Since the embedding of~$\hat B$ onto~$X$ can be permuted in~$|\hat B|! = k!$ ways, the claim follows. By the same argument, taking~$G' = K_k$ (so that~$G = H$ and~$X = \hat B$), we see that~$\aut(H)$ takes the same value, using that weak embeddings of~$H$ into itself are exactly its automorphisms (an embedding between vertex and edge sets of equal size that preserves edges must preserve non-edges too).
\end{proof}

\noindent Accordingly, we have that the number of~$k$-cliques in~$G'$ is exactly~$|\{ X \subseteq V(G') \mid \Psi_X \neq \emptyset\}| = \Count(H \embedsweak{} G) / \aut(H)$ and it follows that counting weak containments of patterns in~$d$-degenerate graphs is indeed~$\# W[1]$-hard. \qed
\end{proof}

\noindent
The characterisation $(1, d)$-locatable patterns (Lemma~\ref{lemma:1-locatable-structure}) poses the
question whether a similar characterisation holds for $(1, d)$-locatable subsets $S'\subseteq S$. In one direction this clearly holds, \ie if a subset~$S'$ has sufficiently many neighbours in~$T$ then one of them must be placed to the right of~$S'$ in any $d$-degenerate ordering:

\begin{corollary}\label{cor:1-locatable-set}
    If~$S' \subseteq S$ has at least~$d+1 - \delta(H[S'])$ joint neighbours or
    if~$H[S']$ is a clique then~$S'$ is $(1,d)$-locatable.
\end{corollary}

\noindent 
However, this is not the only way in which a subset can be located.

\begin{wrapfigure}{l}{0.44\textwidth}
\begin{tikzpicture}[scale=.91]
    \node[vertex] (s1) at (1,0) {$s_1$};
    \node[vertex] (s2) at (2,0) {$s_2$};
    \node[vertex] (s3) at (3,0) {\phantom{$s_3$}};
    \node[vertex] (s4) at (4,0) {\phantom{$s_4$}};

    \node[vertex] (t1) at (1,-1.5) {$t_1$};
    \node[vertex] (t2) at (2,-1.5) {$t_2$};
    \node[vertex] (t3) at (3,-1.5) {$t_3$};
    \node[vertex] (t4) at (4,-1.5) {$t_4$};
    \node[vertex] (t5) at (5,-1.5) {$t_5$};

    \draw[black edge] (s1) -- (t1);
    \draw[black edge] (s1) -- (t2);
    \draw[black edge] (s2) -- (t1);
    \draw[black edge] (s2) -- (t2);

    \draw[black edge] (t3) -- (s1);
    \draw[black edge] (t3) -- (s3);
    \draw[black edge] (t3) -- (s4);
    \draw[black edge] (t4) -- (s2);
    \draw[black edge] (t4) -- (s3);
    \draw[black edge] (t4) -- (s4);
    \draw[black edge] (t5) -- (s2);
    \draw[black edge] (t5) -- (s3);
    \draw[black edge] (t5) -- (s4);    

    \node[right of=s4, node distance=1cm] {$S$};
    \begin{scope}[on background layer]
        \node[box,fit=(s1)(s2)] (Sprime) {};
        \node[left of=Sprime, node distance=1.5cm, color=cardinal] {$S'$};
        \node[box,fit=(t1)(t2),color=gray] (A) {};
        \node[left of=A, node distance=1.5cm, color=gray] {$A$};
    \end{scope}
\end{tikzpicture}
% \vspace*{-2\baselineskip}
\end{wrapfigure}
\noindent The set~$S'$ in the example on the left is a $(1,2)$-locatable subset in a $2$-degenerate pattern.
The minimum degree inside~$S'$ is zero, hence we would expect to need at least~$3$ vertices in~$T$ to locate it according to the bound provided by Corollary~\ref{cor:1-locatable-set}. However, note that any $2$-degenerate ordering must place one vertex of the marked set~$A$ as the rightmost vertex (since they are the only vertices of degree two). Accordingly, $S'$ is $(1,2)$-locatable despite $A$ being smaller than the bound.

Let us now briefly turn our attention to larger values of~$c$.
One easy way of constructing a $(c,d)$-locatable pattern is by covering~$S$ with~$c$ sets that are each~$(1,d)$-locatable. This `naturally' occurs when there are sufficiently many vertices in~$T$ with pairwise distinct neighbourhoods and large enough degree:

\begin{lemma}
    Let~$H$ be a pattern, $d \in \N$ and $p \leq |S| - d$ such that in~$H$ 
    all subsets~$X \in {S \choose p}$ have at least one dedicated joint neighbour~$t_X$, meaning for any two distinct sets~$X_1, X_2 \in {S \choose p}$ we have~$t_{X_1} \neq t_{X_2}$. Then~$H$ is~$(\lceil |S| / (p-1) \rceil,d)$-locatable.
  \end{lemma}
\begin{proof}
	Consider a subset~$S' \subseteq S$ of size~$p-1$. Then~$S'$ has 
	at least
	$
	|S \setminus S'| = |S|-p+1
	$
	joint neighbours among the vertices $\{t_X\}_{X \in {S \choose p}}$.
	Since~$p \leq |S| - d$ this translates to the set~$S'$ having at least~$d+1$ neighbours, which makes it~$(1,d)$-locatable. Since we can cover~$S$
	with~$\lceil |S| / (p-1) \rceil$ sets of size~$p-1$, the claim follows. \qed
\end{proof}

\noindent
Since~$(c,d)$-locatable patterns are somewhat similar to graphs with domination number~$c$, we do not expect a simple characterisation for~$(c,d)$-locatable graphs for larger values of~$c$.

\section{Beyond simple locatability}
\label{sec:beyond-locatability}

Let us briefly outline that the approach outlined in Section~\ref{sec:algorithm} can be improved in certain cases which hints at the possibility of a useful classification of patterns that goes beyond $(c,d)$-locatability. 

\begin{wrapfigure}{r}{0.4\textwidth}
   \vspace*{-0.5\baselineskip}  
  \centering
    \includegraphics[width=.4\textwidth]{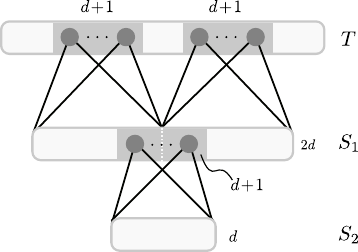}
  \caption{\small A pattern that benefits from `progressive' locating.}
  \label{fig:progressive}
  \vspace*{-2.3\baselineskip}  
\end{wrapfigure}

Consider the following example pattern with vertex sets~$T$, $S_1$, $S_2$ with the following properties (see Figure~\ref{fig:progressive}):
\begin{itemize}
    \item $S_1$ has size~$2d$ and $S_2$ has size~$d$,
    \item there exist two disjoint bicliques $K_{d+1,d}$ between $T$ and~$S_1$
          with the larger side in~$T$ and the two smaller sides together covering~$S_1$, and
    \item there exists a $K_{d+1,d}$ biclique between $S_1$ and~$S_2$ with the 
          larger side in~$S_1$ and the smaller side covering~$S_2$.
\end{itemize}
Other than the above, the edges in~$S_1 \cup S_2$ and between $S_1 \cup S_2$ and $T$ are arbitrary as long as the resulting pattern is still $d$-degenerate. Note that as before, once~$S_1 \cup S_2$ is located we can count the number of strong or weak patterns using the above algorithms, so let us here focus only on the locating part of the algorithm.

By the above construction, the set~$S_1$ is~$(2,d)$-locatable and the set~$S_2$ is~$(1,d)$-locatable. Therefore~$S_1 \cup S_2$ is $(3,d)$-locatable and we can find it in time~$O(|G|^3)$ in a $d$-degenerate host graph. However, we can also locate the set~$S_1 \cup S_2$ as follows: first we guess two vertices in~$T$ which locate~$S_1$ in~$O(|G|^2)$ time and then we guess a vertex from~$S_1$ that locates~$S_2$, taking in total~$O(|G|^2 |S_1|) = O(d|G|^2)$ time, thus improving by a factor of~$O(|G|/d)$.

This example can easily be extended to create patterns that are~$(2c+1,d)$-locatable but allow a `progressive' locatability in time~$O(|G|^{c+1} |S_1|^{c})$ or even `deeper' examples where~$S_1$ helps to locate~$S_2$ which helps locate a further set~$S_3$ \etc. This small toy example clearly demonstrates that the variations of locatability could be of further interest. We leave a formalisation of these ideas for future work.

\section{Conclusion}

We have initiated the study of counting \emph{large} subgraphs in
$d$-degenerate graphs and provided some positive and negative findings related
to $(c,d)$-locatable patterns. We further provided a characterisation of
$(1,d)$-locatable patterns which our algorithms count with only a linear
dependence on~$|G|$. In the future we plan on implementing and testing these
algorithms on real-world networks, as well as investigating what type of patterns might be of interest in practical applications.

On the theoretical side, our main open question is whether $(1,d)$-locatable patterns are the only patterns (modulo small patterns) that can be counted in linear time in $d$-degenerate graphs. Further, we ask whether $(2,d)$-locatable patterns still have a simple characterisation. Finally, we note that our algorithm (Lemma~\ref{lemma:left-cover-compute}) to decide whether a pattern is~$(c,d)$-locatable has a surprisingly bad dependence on~$c$ and~$d$. While for practical applications this should hardly matter, as~$c$ needs to be small anyway, it would be interesting to know whether better dependence on these parameters is possible.

\bibliographystyle{abbrv}
\bibliography{biblio}

\end{document}